\documentclass{fundam}

\publyear{22}
\papernumber{2163}
\volume{190}
\issue{1}

\finalVersionForARXIV

\usepackage{url}
\usepackage{tikz}
\usetikzlibrary{svg.path}
\usepackage{Koprowski}

\begin{document}

\title{Computing Square Roots in Quaternion Algebras}

\author{Przemys{\l}aw Koprowski\thanks{Address for correspondence:  Institute of Mathematics,
           University of Silesia,  ul. Bankowa 14, 40-007 Katowice, Poland. \newline \newline
                    \vspace*{-6mm}{\scriptsize{Received January 2023; \ accepted September 2023.}}}
\\
Institute of Mathematics \\
University of Silesia in Katowice \\
ul. Bankowa 14, 40-007 Katowice, Poland\\
przemyslaw.koprowski@us.edu.pl
}

\maketitle

\runninghead{P. Koprowski}{Computing Square Roots in Quaternion Algebras}

\begin{abstract}
We present an explicit algorithmic method for computing square roots in quaternion algebras over global fields of characteristic different from~$2$.
\end{abstract}


\section{Introduction}
The computation of square roots is one of the most basic operations in mathematics. Effective methods for computing square roots are among the oldest algorithms in the realm of computational mathematics. In fact, Heron's method for a numerical approximation of a square root of a real number is two thousand years old and preceded by the Euclidean algorithm (wildly believed to be the oldest mathematical algorithm) by only about three to four centuries (for an in-depth discussion on the chronology see \cite{Heath1981}). Although numerous methods for computing square roots in various algebraic structures are known nowadays, some important omissions prevail. Among them are general quaternion algebras. Computation of square roots in the algebra of Hamilton quaternions $\HH = \quat{-1}{-1}{\RR}$ is well-known (see \cite{Niven1942}) and very simple as for every quaternion $\qq\in \HH$ there is a subfield $K \cong \CC$ of~$\HH$ containing $\qq$, and so the computation of the square root in~$\HH$ can be reduced to the computation of the square root in~$\CC$. It is no longer so in a general quaternion algebra $\CQ = \quat{\alpha}{\beta}{K}$ for an arbitrary field~$K$ and two elements $\alpha, \beta\in \un$. To the best of our knowledge, no algorithm for computing quaternionic square roots exists in the literature. One possible explanation for this (quite surprising) fact is that in the commutative case when one considers a field extension $\ext{L}{K}$, a typical way to compute a square root of an element $a\in L$ is to factor the polynomial $x^2 - a$ in $L[x]$. However, for quaternion algebras, there are no known polynomial factorization algorithms.

The sole purpose of this paper is to correct this evident omission and present an explicit algorithm for computing square roots in quaternion algebras over arbitrary global fields of characteristic different from~$2$.

\section{Notation}
Throughout this paper, $K$ will denote an arbitrary global field of characteristic $\chr K\neq 2$. Hence, $K$ is either a number field, i.e. a finite extension of~$\QQ$ (then its characteristic is just~$0$) or a global function field, that is, a finite extension of a rational function field over a finite field~$\FF_q$, where $q$ is a power of an odd prime. The set of nonzero elements of~$K$ is denoted~$\un$.

\medskip
Recall that a quaternion algebra $\CQ = \quat{\alpha}{\beta}{K}$ over~q
$K$ is a $4$-dimensional $K$-algebra with a basis $\{1, \ii, \jj, \kk\}$ and a multiplication gathered by the rules:
\[
\ii^2 = \alpha,\quad
\jj^2 = \beta,\quad
\ii\jj = \kk = -\jj\ii.
\]
As usual, we shall identify the field~$K$ with the subfield $K\cdot 1$ of~$\CQ$, which is known to coincide with the center $Z(\CQ)$ of~$\CQ$. We refer the reader to \cite{Vigneras1980, Voight2021} for a comprehensive presentation of the theory of quaternion algebras.

\medskip
A quaternion $\qq$ is called \term{pure} (see e.g., \cite[Definition~5.2.1]{Voight2021}) if $\qq\in \lin\{\ii, \jj, \kk\}$. Every quaternion $\qq\in \CQ$ can be uniquely expressed as a sum $\qq = a + \qq_0$ of a scalar $a\in K$ and a pure quaternion~$\qq_0$. We write $\conj{\qq} := a - \qq_0$ for the \term{conjugate} of~$\qq$. The map that sends a quaternion to its conjugate is an involution.

If $x$ is an element of either a quadratic field extension $L = K\bigl(\sqrt{\alpha}\bigr)$ of~$K$ or a quaternion algebra $\CQ = \quat{\alpha}{\beta}{K}$ over~$K$, we write $\norm[]{x} := x\conj{x}$ and call it the \term{norm} of~$x$. If the domain is not clear from the context, we write $\norm[\ext{L}{K}]{}$ or $\norm[\ext{\CQ}{K}]{}$.

\begin{remark}
When $\CQ$ is a quaternion algebra, the norm of~$\qq$ in the above sense should not be confused with the determinant of the endomorphism of~$\CQ$ defined by the multiplication by~$\qq$, which is often also called the norm. For this reason, in \cite{Vigneras1980, Voight2021} the map $\qq\mapsto \qq\conj{\qq}$ is called the \term{reduced norm} and denoted $\operatorname{nrd}$. In that manner, our terminology in the present paper agrees with the one used by Lam in \cite{Lam2005} but not with the one used by Vigneras in \cite{Vigneras1980} and Voight in \cite{Voight2021}.
\end{remark}

Equivalency classes of valuations on~$K$ are called \term{places}. Throughout this paper, places are denoted using fraktur letters $\gp$, $\gq$, $\gr$. Every place of a global field is either \term{archimedean}, when it extends the standard absolute value on~$\QQ$ (then the field~$K$ is necessarily a number field) or \term{non-archimedean}. Over a global function field, every place is non-archimedean. To avoid monotonous repetitions, non-archimedean places will also be called \term{primes} (or \term{finite primes} when we want to emphasize the fact that they are non-archimedean). The completion of~$K$ with respect to a place~$\gp$ is denoted~$\Kp$. If $\gp$ is a finite prime, we write $\ord_\gp : K \to \ZZ$ to denote the corresponding (normalized) discrete valuation on~$K$. The prime $\gp$ is called dyadic if $\ord_\gp 2 \neq 0$. The map $\ord_\gp$ induces a natural map $\sqgp \to \quo{\ZZ}{2\ZZ}$ on the group of square classes of~$K$ that is again denoted $\ord_\gp$.

If $\gp$ is an archimedean place, then the completion~$\Kp$ is isomorphic either to $\CC$ or to~$\RR$. The places of the second kind are called \term{real}. The field~$K$ is \term{formally real} if $-1$ is not a sum of squares in~$K$. Otherwise, it is called \term{non-real}. It is well known that a global field is formally real if and only if it contains at least one real place. We write $\sgn_\gr a$ for the sign of $a\in K$ with respect to the unique ordering of~$K$ induced by a real place~$\gr$.

\medskip
Given some nonzero elements $a_1, \dotsc, a_n\in K$ we denote by $\form{a_1, \dotsc, a_n}$ the quadratic form $a_1x_1^2 + \dotsb + a_nx_n^2$. Further, if $\gp$ is a place and $a, b\in \un$ we write $(a, b)_\gp$ for the Hilbert symbol of~$a$ and $b$ at~$\gp$, that is
\[
(a, b)_\gp :=
\begin{cases}
1 &\text{if $\quat{a}{b}{\Kp}\cong M_2\Kp$,}\\
-1 &\text{otherwise.}
\end{cases}
\]
Here, $M_2\Kp$ denotes the ring of $2\times 2$ matrices with entries in~$\Kp$. The Hilbert symbol is symmetric, bi-multiplicative and for every $a, b\in \un$ and every place~$\gp$ one has $(a, b)_\gp\cdot (a, b)_\gp = 1$. These three properties of the Hilbert symbol will be extensively used in the paper.

\medskip
For a quadratic form $\xi = \form{a_1, \dotsc, a_n}$ we define its Hasse invariant $s_\gp\xi$ at~$\gp$ by the formula (see e.g., \cite[Definition~V.3.17]{Lam2005}):
\[
s_\gp\xi :=
\prod_{i < j} (a_i, a_j)_\gp.
\]
Given a quadratic form $\xi = \form{a_1, \dotsc, a_n}$ over~$K$ and a place~$\gp$, we write $\xi\otimes \Kp$ for the form over~$\Kp$ with the same entries as~$\xi$. If $\gr$ is a real place, the form $\xi\otimes K_\gr$ is called \term{definite} if all its entries $a_1, \dotsc, a_n$have the same sign. Otherwise, it is called \term{indefinite}.

\medskip
Finally, abusing the notation harmlessly, by $\lo$ we will denote the (unique) isomorphism from the multiplicative group $\{\pm 1\}$ to the additive group $\{0,1\}$ with addition modulo~$2$.

\section{Square roots of non-central elements}
Let us begin by writing down the explicit formula for a square in quaternion algebra so that we can easily reference it in the discussion that follows.

\begin{obs}
If $\qq = q_0 + q_1\ii + q_2\jj + q_3\kk\in \CQ$  is a quaternion, then
\begin{equation}\label{eq:q^2}
\begin{split}
\qq^2
&= (q_0^2 + q_1^2\alpha + q_2^2\beta - q_3^2\alpha\beta)
    + 2q_0q_1\ii + 2q_0q_2\jj + 2q_0q_3\kk
\\
&= ( 2q_0^2 - N(\qq) ) + 2q_0\cdot (q_1\ii + q_2\jj + q_3\kk).
\end{split}
\end{equation}
\end{obs}

An immediate consequence of the previous observation is the following rather well-known fact.

\begin{corollary}\label{cor:square_of_pure}
If $\qq\in \CQ$ is a pure quaternion, then $\qq^2\in Z(\CQ)= K$.
\end{corollary}

Another direct consequence of Eq.~\eqref{eq:q^2} is the following observation that may be treated as a partial converse of Corollary~\ref{cor:square_of_pure}.

\begin{obs}\label{obs:sqrt_in_K}
Let $\qq\in \CQ$ be a square root of some element $a\in K$. Then $\qq$ is either pure or $\qq\in K$.
\end{obs}

\begin{proof}
Let $\qq = q_0 + q_1\ii + q_2\jj + q_3\kk$. If $\qq^2 = a\in K$ then by Eq.~\eqref{eq:q^2} we have
\[
2q_0q_1 = 2q_0q_2 = 2q_0q_3 = 0.
\]
Therefore, if $\qq$ is not pure, that is if $q_0 \neq 0$, then $q_1 = q_2 = q_3 = 0$, hence $\qq\in K$.
\end{proof}

Combining Corollary~\ref{cor:square_of_pure} with Observation~\ref{obs:sqrt_in_K} we see that for computing the square roots in quaternion algebras it is crucial to distin­guish between the case when one computes a quaternionic square root of an element in~$K$ (i.e., in the center of~$\CQ$)  and the case when the argument comes from $\CQ\setminus Z(\CQ)$. It turns out that the latter case is, in fact, trivial and requires nothing more than high-school mathematics.

\begin{alg}
Let $\CQ = \quat{\alpha}{\beta}{K}$ be a quaternion algebra over a field~$K$ of characteristic $\chr K\neq 2$. Given a quaternion $\qq = q_0 + q_1\ii + q_2\jj + q_3\kk\in \CQ\setminus Z(\CQ)$, this algorithm outputs its square root or reports a failure when $\qq$ is not a square.
\begin{enumerate}
\itemsep=0.95pt
\item Check if the norm $N(\qq)$ of~$\qq$ is a square in~$K$.
  \begin{enumerate}
  \item\label{st:sqrt_not_pure:early} If it is not, then report a failure and quit.
  \item If it is, let $d$ be an element of~$K$ such that $d^2 = N(\qq)$.
  \end{enumerate}
\item Check if any of the following two elements is a square in~$K$:
\[
a_+ := \frac{q_0 + d}{2},\qquad a_- := \frac{q_0 - d}{2}.
\]
\item If neither of them is a square, then report a failure and quit.
\item Otherwise, fix $r_0$ such that either $r_0^2 = a_+$ or $r_0^2 = a_-$.
\item Set
\[
r_1 := \frac{q_1}{2r_0},\quad
r_2 := \frac{q_2}{2r_0},\quad
r_3 := \frac{q_3}{2r_0}.
\]
\item Output $\rr = r_0 + r_1\ii + r_2\jj + r_3\kk$.
\end{enumerate}
\end{alg}

\begin{poc}
Since the norm $N:\CQ\to K$ is multiplicative, it is obvious that if $N(\qq)\notin K^2$, then $\qq$ cannot be a square in~$\CQ$. This fact justifies the early exit in step~\eqref{st:sqrt_not_pure:early} of the algorithm. Assume that $N(\qq) = d^2$ and let $\rr = r_0 + r_1\ii + r_2\jj + r_3\kk$ be the sought square root of~$\qq$, if it exists. By Eq.~\eqref{eq:q^2} we have
\[
q_1 = 2r_0r_1,\qquad
q_2 = 2r_0r_2,\qquad
q_3 = 2r_0r_3.
\]
It is, thus, clear that it suffices to find $r_0$. Again by Eq.~\eqref{eq:q^2} we may write
\[
q_0
= r_0^2 + r_1^2\alpha + r_2^2\beta - r_3^2\alpha\beta
\\
= r_0^2 + \Bigl(\frac{q_1}{2r_0}\Bigr)^2\alpha + \Bigl(\frac{q_2}{2r_0}\Bigr)^2\beta - \Bigl(\frac{q_3}{2r_0}\Bigr)^2\alpha\beta.
\]
\eject
The above formula can be rewritten in the form of a bi-quadratic equation:
\[
4r_0^4 - 4q_0r_0^2 + \bigl(q_1^2\alpha + q_2^2\beta - q_3^2\alpha\beta\bigr) = 0.
\]
If we treat the left-hand-side as a quadratic equation in $r_0^2$, then its discriminant equals $16\cdot N(\qq) = (4d)^2$, hence
\[
r_0^2 = \frac{q_0 \pm d}{2} = a_\pm.
\]
It follows that the sought quaternion~$\rr$ exists if and only if either $a_+$ or $a_-$ is a square in~$K$. This proves the correctness of the algorithm.
\end{poc}

\begin{remark}\label{rem:num_of_sqrt}
In the above proof, we constructed the square root~$\rr$ of a quaternion $\qq\in \CQ\setminus Z(\CQ)$ by solving a bi-quadratic equation. Such equations in general, may have four roots. Hence, one may suspect that there are four distinct quaternions~$\rr$ such that $\rr^2 = \qq$. It is not the case. It is clear from the above proof that $\qq\in \CQ\setminus Z(\CQ)$ has only finitely many square roots in~$\CQ$. Now, if $\rr^2\in \QQ$, then $\rr$ is a root of a quaternionic polynomial $x^2 - \qq$. But \cite[Theorem~5]{GM1965} asserts that a quadratic polynomial over~$\CQ$ which has more than two zeros must have infinitely many of them. This way, we conclude that $\qq$ has just two square roots. Notice that for hamiltonian quaternions this fact has been observed already 80 years ago by Niven in \cite{Niven1942}.
\end{remark}

\section{Square roots of central elements. Split case}
It is evident from the preceding section that the only non-trivial case that must be considered is how to compute a quaternionic square root of an element of the base field~$K$, which is not a square in~$K$. In contrast to the previous case (cf. Remark~\ref{rem:num_of_sqrt}), in general, an element $a\in K = Z(\CQ)$ may have infinitely many square roots in~$\CQ$. Once again, for hamiltonian quaternions it has been observed already by Niven.

First, we need, however, to introduce an auxiliary algorithm that is not specific to quaternions, as it deals with an arbitrary quadratic form. Recall that a quadratic form is called \term{isotropic} (see e.g., \cite[Definition~I.3.1]{Lam2005}) if it represents zero non-trivially. It is well known (see, e.g., \cite[Theorem~I.3.4]{Lam2005}) that every isotropic form represents all elements of~$K$.

\begin{alg}\label{alg:isotropic2universal}
Let $\xi$ be an isotropic quadratic form of dimension~$n$ over a field~$K$ of characteristic $\chr K\neq 2$. Given an element $a\in K$ and a vector $V\in K^n$ such that $\xi(V) = 0$, this algorithm outputs a vector $W\in K^n$ satisfying the condition $\xi(W) = a$.
\begin{enumerate}
\itemsep=0,9pt
\item Find a vector $U\in K^n$ such that $U$ and $V$ are linearly independent.
\item Set $b := \xi(U)$ and $c := \sfrac12\cdot \bigl(\xi(U + V) - \xi(U)\bigr)$.
\item Output
\[
W := U + \frac{a - b}{2c}\cdot V.
\]
\end{enumerate}
\end{alg}

\begin{poc}
Just compute:
\begin{eqnarray*}
\xi(W)
&=& \xi\Bigl(U + \frac{a-b}{2c}\cdot V\Bigr)
\\[2pt]
&=& \xi(U) + \frac{a-b}{2c}\cdot \bigl(\xi(U + V) - \xi(U) - \xi(V)\bigr) + \frac{(a-b)^2}{4c^2}\xi(V)
\\[2pt]
&=& b + \frac{a - b}{2c}\cdot 2c + 0
= a
\end{eqnarray*}

\vspace*{-8mm}
\end{poc}

Recall that a quaternion algebra $\CQ = \quat{\alpha}{\beta}{K}$ is said to \term{split} (see e.g., \cite[Definition~5.4.5]{Voight2021}) if $\CQ$ is isomorphic to the matrix ring $M_2K$. It is well known (see e.g., \cite[Theorem~5.4.4]{Voight2021} or \cite[Theorem~III.2.7]{Lam2005}) that $\CQ$ is split if and only if the quadratic form $\form{-\alpha, -\beta, \alpha\beta}$ is isotropic. If it is the case, the preceding algorithm combined with Eq.~\eqref{eq:q^2} lets us compute the quaternionic square root of any element of the base field. In particular, when $K$ is a global field, $\chr K\neq 2$, then the computation of the square root of $a\in K$ in a split quaternion algebra boils down to solving a norm equation in a quadratic extension of~$K$. Algorithms for the latter task are well known. They can be found in \cite{Cohen2000, FJP1997, FP1983, Garbanati1980, Simon2002}.

\begin{alg}\label{alg:sqrt_split}
Let $\CQ = \quat{\alpha}{\beta}{K}$ be a split quaternion algebra over a global field~$K$ of characteristic $\chr K\neq 2$. Given a nonzero element $a\in K$, this algorithm outputs a pure quaternion $\qq\in \CQ$ such that $\qq^2 = a$.
\begin{enumerate}
\itemsep=0.95pt
\item\label{st:sqrt_split:v1} Check if $\alpha$ is a square in~$K$. If there is $c\in \un$ such that $c^2 = \alpha$, then set $V := (0, c, 1)$.
\item\label{st:sqrt_split:v2} Otherwise, if $\alpha \notin \sq$, then:
  \begin{enumerate}
  \item Construct a quadratic field extension $L = K\bigl(\sqrt{\alpha}\bigr)$ of~$K$.
  \item\label{st:sqrt_split:norm} Solve the norm equation
  \[
  \norm[\ext{L}{K}]{x} = -\frac{\alpha}{\beta}
  \]
  and denote the solution by $\lambda = b + c\sqrt{\alpha}$.
  \item Set $V := (1, b, c)$.
  \end{enumerate}
\item\label{st:sqrt_split:w} Let $\xi := \form{-\alpha, -\beta, \alpha\beta}$ be the pure subform of the norm form of~$\CQ$. Execute Algorithm~\ref{alg:isotropic2universal} with the input $(-a, V, \xi)$ to construct a vector $W = (w_1, w_2, w_3)$ such that $\xi(W) = -a$.
\item Output $\qq = 0 + w_1\ii + w_2\jj + w_3\kk$.
\end{enumerate}
\end{alg}

\begin{poc}
We claim that the vector~$V$ constructed either in step~\eqref{st:sqrt_split:v1} or in step~\eqref{st:sqrt_split:v2} of the algorithm is an isotropic vector for $\xi$. First, suppose that $\alpha$ is a square in~$K$. Say $\alpha = c^2$ for some $c\in \un$. Then
\[
-\alpha\cdot 0^2 - \beta\cdot c^2 + \alpha\beta\cdot 1^2 = 0.
\]
Conversely, assume that $\alpha\notin \sq$ and so $L = K\bigl(\sqrt{\alpha}\bigr)$ is a proper extension of~$K$. Let $\lambda = b + c\sqrt{\alpha}$ be an element of~$L$ such that $\norm\lambda = -\sfrac\alpha\beta$. Then
\[
-\frac\alpha\beta = \lambda\conj\lambda = b^2 - \alpha c^2.
\]
It follows that
\[
-\alpha\cdot 1^2 - \beta\cdot b^2 + \alpha\beta\cdot c^2 = 0.
\]
Hence, in both cases $V$ is an isotropic vector of~$\xi$, as claimed. Consequently, executing Algorithm~\ref{alg:isotropic2universal} in step~\eqref{st:sqrt_split:w} we obtain a vector~$W$ satisfying the condition $\xi(W) = -a$. Now, by Eq.~\eqref{eq:q^2} the square of the quaternion~$\qq$ outputted by the algorithm equals
\[
\qq^2 = -N(\qq) = -\xi(W) = a.
\]
Thus, to conclude the proof, we only need to show that the norm equation in step~\eqref{st:sqrt_split:norm} is solvable. But this follows immediately from the fact that $\CQ$ is split. Hence $\xi$ is isotropic. Indeed, if $V = (v_1, v_2, v_3)$ is an isotropic vector of~$\xi$, then
\[
-\alpha\cdot v_1^2 - \beta\cdot v_2^2 + \alpha\beta\cdot v_3^2 = 0.
\]
Observe that $v_1$ must be nonzero since otherwise, $\alpha$ would be a square. It follows that
\[
-\frac\alpha\beta
= \Bigl(\frac{v_2}{v_1}\Bigr)^2 - \alpha \Bigl(\frac{v_3}{v_1}\Bigr)^2
= \norm[\ext{L}{K}]{ \frac{v_2}{v_1} + \frac{v_3}{v_1}\sqrt\alpha }.
\]
Therefore, the norm equation is solvable, as claimed.
\end{poc}

\begin{remark}
The construction of the isotropic vector~$V$ in steps~(\ref{st:sqrt_split:v1}--\ref{st:sqrt_split:v2}) of Algorithm~\ref{alg:sqrt_split} is equivalent to establishing an explicit isomorphism $\CQ\cong M_2K$. For details, see \cite[Chapter~III]{Lam2005}. Of course, if the quaternion algebra~$\CQ$ is fixed, the vector~$V$ should be computed only once and cached between successive computations of square roots.
\end{remark}

\begin{remark}
If the isomorphism $\CQ\cong M_2K$ is a priori known explicitly, then the computation of the quaternionic square root of any $a\in \un$ trivializes, as we have the identity
\[
\begin{pmatrix} 0 & a\\ 1 & 0\end{pmatrix}^2
=
\begin{pmatrix} a & 0\\ 0 & a\end{pmatrix}.
\]
\end{remark}

\section{Square roots of central elements. Non-split case}
Now the only case left to be dealt with is when $a\in \un$ but $\CQ$ is not split. Here we have to solve not one but two norm equations (see Algorithm~\ref{alg:nonsplit_sqrt} below). First, however, we need to introduce the following auxiliary algorithm that constructs an element simultaneously represented by two binary forms. Recall (see e.g., \cite[Definition~I.2.1]{Lam2005}) that for a given quadratic form~$\xi$ of dimension~$d$, we denote the set of nonzero elements of~$K$ represented by $\xi$ by the symbol
\[
D_K(\xi) := \bigl\{ \xi(V)\mid V\in K^d\text{ and }\xi(V)\neq 0 \bigr\}.
\]

Let $\GP$ be any finite set of primes of~$K$. Recall that an element $a\in \un$ is called \term{$\GP$-singular} if $\ord_\gp a\equiv 0\pmod{2}$ for all finite primes $\gp\notin \GP$. The set of all $\GP$-singular elements forms a subgroup of the group $\un$ containing $\sq$. Thus, the notion of $\GP$-singularity generalizes naturally to the square classes. Define the set
\[
\SingS := \bigl\{ a\sq\st a\text{ is $\GP$-singular} \bigr\}
\]
of $\GP$-singular square classes. It is a subgroup of the group $\sqgp$ of square classes of~$K$, hence a vector space over~$\FF_2$. It is known that the dimension of this vector space is finite. In fact it equals (see e.g., \cite[p.~607]{Czogala2001})
\[
\dim_{\FF_2}\SingS = \card{\GP} + \dim_{\FF_2} \quo{C_\GP}{C_\GP^2},
\]
where $C_\GP$ is the $\GP$-class group of~$K$. There is a number of known algorithms to construct a basis of this vector space. For details see e.g., \cite{CBFS2015, Koprowski2021, KR2023}.

\medskip
Before we present the next algorithm it is crucial to point out that the set of non-archimedean places of a global field~$K$ is countable. All the places of~$K$ can be arranged into an infinite sequence $\gq_1, \gq_2,\dotsc$. One possible way to do that is the following one. If $K$ is a number field, let $p_1, p_2, p_3, \dotsc = 2, 3, 5, \dotsc$ be the (strictly increasing) sequence of all prime numbers. On the other hand, if $K$ is a global function field, i.e., a finite extension of a rational function field $\FF_q(x)$, let $p_1 = \sfrac1x$ and $p_2, p_3, p_4, \dotsc$ be a sequence of all the irreducible polynomials from $\FF_q[x]$ ordered in such a way that $\deg p_j\leq \deg p_{j+1}$ for every~$j$. Now, we can first take the places of~$K$ that extend $p_1$, then the ones that extend $p_2$, then $p_3$, and so on. Consequently, it is possible to iterate over the set of primes of~$K$. This observation will be indispensable for the rigorous proof of correctness of the algorithm that follows.

\begin{alg}\label{alg:intersection}
Let $K$ be a global field of characteristic $\chr K\neq 2$. Given two binary quadratic forms $\xi = \form{x_0, x_1}$ and $\zeta = \form{z_0, z_1}$ over~$K$ with $x_0, x_1, z_0, z_1 \neq 0$, this algorithm outputs a nonzero element $d\in \un$ such that $d\in D_K(\xi)\cap D_K(\zeta)$ or reports a failure if there is no such~$d$.
\begin{enumerate}
\item\label{st:intersection:trivial} If $-x_0x_1$ is a square in~$K$, then output $z_0$ and quit.
\item\label{st:intersection:trivial2} Likewise, if $-z_0z_1$ is a square in~$K$, then output $x_0$ and quit.
\item\label{st:intersection:early_exit} Check \textup(using e.g., \cite[Algorithm~5]{KCz2018}\textup) whether the form
\[
\xi\perp (-\zeta) = \form{x_0, x_1, -z_0, -z_1}
\]
is isotropic. If it is not, then report a failure and quit.
\item\label{st:intersection:GP} Construct a set $\GP$ consisting of all dyadic places of~$K$ \textup(if there are any\textup) and of all these non-dyadic primes of~$K$ where at least one of the elements $x_0, x_1, z_0, z_1$ has an odd valuation.
\item\label{st:intersection:GR} If $K$ is a formally real number field, then:
  \begin{enumerate}
  \item Construct the set~$\GR$ of all the real places of~$K$, where either $\xi$ or $\zeta$ is definite and denote its cardinality by~$r$, i.e.
  \[
  \GR = \bigl\{ \gr \st \sgn_\gr x_0x_1 = 1
  \text{ or }
  \sgn_\gr z_0z_1 = 1 \bigr\},
  \qquad
  r = \card{\GR}.
  \]
  \item{} \textup[Notation only\textup] Let $\gr_1, \dotsc, \gr_r$ be all the elements of~$\GR$.
  \item Construct a vector $W = (w_1, \dotsc, w_r)\in \{0,1\}^r$ setting
  \[
  w_i =
  \begin{cases}
  \lo \sgn_{\gr_i}x_0 & \text{if }\sgn_{\gr_i}x_0x_1 = 1,\\
  \lo \sgn_{\gr_i}z_0 & \text{if }\sgn_{\gr_i}x_0x_1 = -1.
  \end{cases}
  \]
  \end{enumerate}
Otherwise, if the field~$K$ is non-real, set $\GR := \emptyset$, $r = 0$ and $W := ()$.
\item Repeat the following steps until the sought element~$d$ is found:
  \begin{enumerate}
  \item{} \textup[Notation only\textup] Let $\gp_1, \dotsc, \gp_s$ be all the elements of~$\GP$.
  \item Construct a basis $\SB = \{\beta_1, \dotsc, \beta_k\}$ of the group~$\SingS$ of $\GP$-singular square classes.
  \item\label{st:intersection:UV} Construct vectors $U = (u_1, \dotsc, u_s)$ and $V = (v_1, \dotsc, v_s)$ setting
  \[
  u_i = \lo (x_0, x_1)_{\gp_i}
  \qquad\text{and}\qquad
  v_i = \lo (z_0, z_1)_{\gp_i}.
  \]
  \item\label{st:intersection:AB} Construct matrices $A = (a_{ij})$ and $B = (b_{ij})$, with $k = \card{\SB}$ columns and $s = \card{\GP}$ rows, setting
  \[
  a_{ij} = \lo ( -x_0x_1, \beta_j)_{\gp_i}
  \qquad\text{and}\qquad
  b_{ij} = \lo ( -z_0z_1, \beta_j)_{\gp_i}.
  \]
  \item If $\GR\neq \emptyset$ construct a matrix $C = (c_{ij})$ with $k$ columns and $r = \card{\GR}$ rows, setting
  \[
  c_{ij} = \lo \sgn_{r_i} \beta_j.
  \]
  Otherwise, when $\GR = \emptyset$, set $C = ()$.
  \item Check if the following system of $\FF_2$-linear equations has a solution 
  {
  \newlength{\tagmarginsep}
  \setlength{\tagmarginsep}{2cm}
  \everydisplay{\displayindent=\tagmarginsep \displaywidth=\dimexpr\linewidth-2\tagmarginsep}
  \begin{equation}\tag{$\Qoppa$}\label{eq:intersection}
  \mbox{}\hspace{15mm}\left(\begin{array}{c}
  A\\ \hline B\\\hline C
  \end{array}\right)
  \cdot
  \begin{pmatrix}
  x_1 \\ \vdots\\ x_k
  \end{pmatrix}
  =
  \left( \begin{array}{c}
         U\\
         \hline V\\
         \hline W
  \end{array}\right)
  \end{equation}
  }
  \item If it does, denote the solution by $(\varepsilon_1, \dotsc, \varepsilon_k)\in \{0, 1\}^k$. Output $d = \beta_1^{\varepsilon_1}\dotsb \beta_k^{\varepsilon_k}$ and quit.
  \item\label{st:intersection:new_prime} If the system~\eqref{eq:intersection} has no solution, then append to~$\GP$ the first prime~$\gq_j$ of~$K$ that is not yet in~$\GP$ \textup(see the comment preceding the algorithm\textup) and reiterate the loop.
  \end{enumerate}
\end{enumerate}
\end{alg}

\begin{poc}
First, suppose that $-x_0x_1$ is a square in~$K$. This means that the form~$\xi$ is isotropic (see, e.g., \cite[Theorem I.3.2]{Lam2005}). Hence, by \cite[Theorem~I.3.4]{Lam2005} it represents every element of~$K$. In particular, it represents~$z_0$. Since $\zeta$ also represents $z_0$ (trivially), step~\eqref{st:intersection:trivial} of the algorithm outputs the correct result. The same argument also applies to step~\eqref{st:intersection:trivial2}, when it is the form~$\zeta$ that is isotropic. It is also clear that the sets $D_K(\xi)$ and $D_K(\zeta)$ of elements represented by $\xi$ and $\zeta$, intersect if and only if $\xi\perp (-\zeta)$ is isotropic. This justifies the test in step~\eqref{st:intersection:early_exit}. Therefore, without loss of generality, for the remainder of the proof, we may assume that $\xi\perp (-\zeta)$ is isotropic while both forms $\xi$ and $\zeta$ are anisotropic.

We will first show that the algorithm terminates. Let $W = (w_0, w_1, w_2, w_3)\in K^4$ be an isotropic vector of $\xi\perp (-\zeta)$. Denote $e: = \xi(w_0, w_1) = \zeta(w_2, w_3)$. Further, let $\GR$ and $\GP$ be the sets of places (real and non-archimedean, respectively) constructed in steps (\ref{st:intersection:GP}--\ref{st:intersection:GR}) of the algorithm. We shall now apply \cite[Lemma~2.1]{LW1992}. In the notation of \cite{LW1992} we take $S$ to be the union of~$\GP$ and the set of all non-archimedean places of~$K$. In particular, $S$ contains~$\GR$. For every prime $\gp\in \GP$ we set $n(\gp) := 1 + \ord_\gp 4$. For real places $\gr\in \GR$, take $n(\gr) := 1$. For non-archimedean places $\gp\notin \GR$, the choice of $n(\gp)$ is irrelevant. As in \cite{LW1992}, let $\mathfrak{m} = \prod_{\gp\in S} \gp^{n(\gp)}$ be a modulus. Next, take $b_\gp := e$ for every $\gp\in S$. Moreover, if $K$ is a global function field, pick one more prime not in~$S$ and denote\footnote{This prime is denotes~$p_0$ in \cite{LW1992}, but we will not use this symbol as it would contradict the notation in the rest of the proof.} it~$\gp_*$ and set the corresponding exponent to be~$2$. Then \cite[Lemma~2.1]{LW1992} asserts that there exists a finite prime~$\gp_0$ of~$K$ (denoted~$q$ in \cite{LW1992}) and an element $d\in \un$ (denoted~$b$ ibid) such that:
{
\renewcommand{\theenumi}{\roman{enumi}}
\begin{enumerate}
\itemsep=0.98pt
\item $\ord_\gp d = 0$ for every finite prime $\gp\notin \GP\cup \{\gp_0\}$, except that if $K$ is a global function field, at the singled out prime~$\gp_*$ the valuation $\ord_{\gp_*} d$ is even but possibly nonzero;
\item\label{it:intersection:LST} $d\equiv e\pmod{ \gp^{1 + \ord_\gp 4} }$ for every $\gp\in \GP$;
\item $\ord_{\gp_0} d = 1$;
\item $\sgn_\gr d = \sgn_\gr e$ for every real places~$\gr$ of~$K$.
\end{enumerate}
}
Let $\SB = \{ \beta_1, \dotsc, \beta_k\}$ be a basis of the group $\Sing{\GP\cup \{\gp_0\}}$ of $\bigl(\GP\cup \{\gp_0\}\bigr)$-singular square classes. The element~$d$ is $\bigl(\GP\cup \{\gp_0\}\bigr)$-singular, hence it can be expressed in the form
\[
d = \beta_1^{\varepsilon_1}\dotsm \beta_k^{\varepsilon_k},
\]
where $\varepsilon_1, \dotsc, \varepsilon_k\in \FF_2$ are the coordinates of~$d$ with respect to~$\SB$.

\medskip
Fix a real place $\gr_i\in \GR$. First, suppose that $\sgn_{\gr_i}x_0x_1 = 1$, so the form $\xi\otimes K_{\gr_i}$ is definite. Then $\sgn_{\gr_i} (-d) = \sgn_{\gr_i} (-e) = \sgn_{\gr_i} x_0$ since $\form{-e, x_0, x_1}$ is isotropic. But this implies that
\[
\prod_{j = 1}^k (-1)^{c_{ij}\varepsilon_j}
= \prod_{j = 1}^k \sgn_{\gr_i}\beta_j^{\varepsilon_j}
= \sgn_{\gr_i} d
= \sgn_{\gr_i} x_0
= (-1)^{w_i}.
\]
Consequently
\begin{equation}\label{eq:intersection:C}
c_{i1}\varepsilon_1 + \dotsb + c_{ik}\varepsilon_k = w_i.
\end{equation}
Conversely, assume that $\xi\otimes \Kri$ is indefinite, hence $\zeta\otimes \Kri$ must be definite. Applying the same arguments to the form~$\zeta$ instead of~$\xi$, we show that Eq.~\eqref{eq:intersection:C} also holds in this case.

Now fix a finite prime $\gp_i\in \GP$. Observe that by the local square theorem (see, e.g., \cite[Theorem~VI.2.19]{Lam2005}) condition~\eqref{it:intersection:LST} implies that the local squares classes $d\sq[\Kpi]$ and $e\sq[\Kpi]$ coincide. It follows that the form
\[
\form{ -d, x_0, x_1 }\otimes \Kpi \cong \form{-e, x_0, x_1}\otimes \Kpi
\]
is isotropic. Now, \cite[Proposition~V.3.22]{Lam2005} asserts that the Hasse invariant of $\form{-d, x_0, x_1}\otimes \Kpi$ equals
\begin{equation}\label{eq:Hasse_to_isotropy}
s_{\gp_i}\form{-d, x_0, x_1} = (-1, x_0x_1\cdot d)_{\gp_i}.
\end{equation}
Using the definition of the Hasse invariant and properties of the Hilbert symbol we can rewrite the above condition as follows
\begin{align*}
1
&= s_{\gp_i}\form{-d, x_0, x_1}\cdot (-1, x_0x_1\cdot d)_{\gp_i}
\\
&= (-d, x_0)_{\gp_i} (-d, x_1)_{\gp_i} (x_0, x_1)_{\gp_i} (-1, x_0x_1)_{\gp_i} (-1, d)_{\gp_i}
\\
&= (-d, x_0x_1)_{\gp_i} (-1, x_0x_1)_{\gp_i} (-1, d)_{\gp_i} (x_0, x_1)_{\gp_i}
\\
&= (d, x_0x_1)_{\gp_i} (-1, d)_{\gp_i} (x_0, x_1)_{\gp_i}
\\
&= (-x_0x_1, d)_{\gp_i}(x_0, x_1)_{\gp_i}.
\end{align*}
Therefore, formula~\eqref{eq:Hasse_to_isotropy} is equivalent to the following one:
\[
(-x_0x_1, d)_{\gp_i} = (x_0, x_1)_{\gp_i}.
\]
Substituting $\beta_1^{\varepsilon_1}\dotsm \beta_k^{\varepsilon_k}$ for $d$ we obtain
\[
\prod_{j = 1}^k (-x_0x_1, \beta_j)^{\varepsilon_j}_{\gp_i}
= (x_0, x_1)_{\gp_i}.
\]
Now, $(x_0, x_1)_{\gp_i} = (-1)^{u_i}$ and $(-x_0x_1, \beta_j)_{\gp_i} = (-1)^{a_{ij}}$, where $u_i, a_{ij}\in \{0,1\}$ are the elements constructed in steps (\ref{st:intersection:UV}--\ref{st:intersection:AB}). Therefore, the last condition can be expressed as a linear equation over~$\FF_2$:
\begin{equation}\label{eq:intersection:A}
a_{i1}\varepsilon_1 + \dotsb + a_{ik}\varepsilon_k = u_i.
\end{equation}

Finally, we will show that the above equation also holds for the index $i = 0$, that is for the prime~$\gp_0$ appended to~$\GP$. This fact follows from Hilbert reciprocity law (see, e.g., \cite[Theorem~VI.5.5]{Lam2005}). We already know that for every $i\in \{1,\dotsc, s\}$ we have
\[
( -x_0x_1, d)_{\gp_i} = (x_0, x_1)_{\gp_i}.
\]
The same also holds for primes not in~$\GP$. Indeed, if $\gq\notin \GP\cup \{\gp_0\}$ then $\gq$ is non-dyadic and all three elements $x_0$, $x_1$ and $d$ have even valuations at~$\gq$. Consequently, by \cite[Corollary~VI.2.5]{Lam2005} one obtains
\[
(-x_0x_1, d)_\gq = (x_0, x_1)_\gq = 1.
\]
Now, by Hilbert reciprocity law, we can write
\begin{eqnarray*}
1
&=& \prod_{\gp} (-x_0x_1, d)_\gp\cdot \prod_\gp (x_0, x_1)_\gp
\\
&=& (-x_0x_1, d)_{\gp_0} (x_0, x_1)_{\gp_0}
    \cdot \prod_{\gp \in \GP} \Bigl((-x_0x_1, d)_\gp (x_0, x_1)_\gp\Bigr)
    \cdot \prod_{\mathclap{\gq \notin \GP\cup \{\gp_0\}}} \Bigl((-x_0x_1, d)_\gq (x_0, x_1)_\gq\Bigr)
\\
&=& (-x_0x_1, d)_{\gp_0} (x_0, x_1)_{\gp_0}.
\end{eqnarray*}
Hence, in the same way as above, we show that Eq.~\eqref{eq:intersection:A} also holds for $i = 0$. Applying the same arguments to the form~$\zeta$, we obtain
\begin{equation}\label{eq:intersection:B}
b_{i1}\varepsilon_1 + \dotsb + b_{ik}\varepsilon_k = v_i,
\end{equation}
for all $i \in \{0, 1, \dotsc, s\}$.

\medskip
All in all, we have proved that Eq.~\eqref{eq:intersection} has a solution in $\Sing{\GP\cup \{\gp_0\}}$. Now, for every $\GP'\supseteq \GP\cup \{\gp_0\}$ we have $\Sing{\GP\cup \{\gp_0\}}\subseteq \Sing{\GP'}$, hence once the prime~$\gp_0$ is appended to~$\GP$ the algorithm terminates (see also Remark~\ref{rem:termination} w below).

Now, when we have proved that the algorithm stops, we must show that it outputs a correct result. To this end, we will show that the forms $\form{-d, x_0, x_1}$ and $\form{-d, z_0, z_1}$ are locally isotropic in every completion of~$K$. The assumptions are symmetric with respect to both forms, except in real places. Hence it generally suffices to prove the isotropy of one of them.

Both forms are trivially isotropic in all complex completions of~$K$ (provided that there are any) and in all real completions $K_\gr$ for $\gr\notin \GR$. Fix now a real place $\gr_i\in \GR$. First, assume that the form $\form{x_0, x_1}\otimes \Kri$ is definite. From the preceding part we know that the element $d = \beta_1^{\varepsilon_1}\dotsm \beta_k^{\varepsilon_k}$, constructed by the algorithm, satisfies the condition $\sgn_{\gr_i} d = \sgn_{\gr_i} x_0$. Therefore the form $\form{ -d, x_0, x_1 }\otimes \Kri$ is isotropic. Now, the form $\xi\perp (-\zeta)$ is isotropic because otherwise, the execution of the algorithm would have been interrupted already in step~\eqref{st:intersection:early_exit}. Thus, either $\sgn_{\gr_i} z_0 = \sgn_{\gr_i}x_0 = \sgn_{\gr_i} d$ or $\sgn_{\gr_i} z_1 = \sgn_{\gr_i}x_0 = \sgn_{\gr_i} d$. In both cases, we have that the form $\form{-d, z_0, z_1}\otimes \Kri$ is isotropic, as well.
Conversely, assume that $\xi\otimes \Kri$ is indefinite, and so it is $\zeta\otimes \Kri$ that must be definite. Then, $\form{-d, x_0, x_1}\otimes \Kri$ is trivially isotropic and to the form $\form{-d, z_0, z_1}\otimes \Kri$ we apply the some argument as to the form $\form{-d, x_0, x_1}\otimes \Kri$ in the previous case.

We may now concentrate on finite primes. Fix a prime~$\gp$. Suppose $\gp$ is not among the primes constituting~$\GP$ (here, we allow $\GP$ to have been already enlarged during the execution of the algorithm). In that case, $\gp$ is certainly non-dyadic, and all three elements $x_0$, $x_1$, and $d$ have even valuations at~$\gp$. Hence, \cite[Corollary~VI.2.5]{Lam2005} asserts that $\form{ -d, x_0, x_1 }\otimes \Kp$ is isotropic. On the other hand, we know from the first part of the proof that if $\gp = \gp_i\in \GP$, then $d$ satisfies the condition $(-x_0x_1, d)_\gp = (x_0, x_1)_\gp$, which is equivalent to $s_\gp\form{-d, x_0, x_1} = (-1, x_0x_1\cdot d)_\gp$. The later condition implies that $\form{-d, x_0, x_1}\otimes \Kp$ is isotropic, again by \cite[Proposition~V.3.22]{Lam2005}. The very same arguments may be applied to the form $\form{-d, z_0, z_1}\otimes \Kp$.

All in all, we have shown that the forms $\form{-d, x_0, x_1}$ and $\form{-d, z_0, z_1}$ are locally isotropic in every completion of~$K$. Thus, they are isotropic over~$K$ by the Hasse--Min­kowski principle (see e.g., \cite[Theorem~VI.3.1]{Lam2005}). This means that the forms~$\xi$ and~$\zeta$ represent $d$ over~$K$ by \cite[Corollary~I.3.5]{Lam2005}.
\end{poc}

\begin{remark}\label{rem:termination}
To rigorously prove that Algorithm~\ref{alg:intersection} terminates, we used the fact that it is possible to iterate over the primes of~$K$ arranging all of them into a sequence. Hence, after finitely many steps the prime~$\gp_0$, specified in the proof of correctness, is appended to~$\GP$ and so the algorithm stops. However, it does not present a complete picture. We proved that the corresponding prime~$\gp_0$ exists using \cite[Lemma~2.1]{LW1992}. If one analyzes the proof of this lemma, one will realize that the authors rely on Chebotarev's density theorem to show that the set of primes satisfying the assertions of the lemma has positive density (hence is non-empty, consequently the corresponding prime exists). This means that in a practical implementation, in step~\eqref{st:intersection:new_prime} of the algorithm it is possible to actually add primes to~$\GP$ at random. If the density of the set mentioned above is $d\in (0, 1]$, then the probability that the system~\eqref{eq:intersection} fails to be solvable after $n$ steps is $(1-d)^n$, for sufficiently large~$n$. Hence, it diminishes expotentially with the number of iterations.
\end{remark}

We are now in a position to present an algorithm that computes a square root of a scalar in a non-split quaternion algebra.

\begin{alg}\label{alg:nonsplit_sqrt}
Let $\CQ = \quat{\alpha}{\beta}{K}$ be a non-split quaternion algebra over a global field of characteristic $\chr K\neq 2$. Given a nonzero element $a\in K$ this algorithm outputs a quaternion $\qq\in \CQ$ such that $\qq^2 = a$ or reports a failure if $a$ is not a square in~$\CQ$.
\begin{enumerate}
\itemsep=0.95pt
\item\label{st:nonsplit_sqrt:a_square} Check if $a$ is a square in~$K$. If there is $c \in \un$ such that $a = c^2$, then output $\qq = c + 0\ii + 0\jj + 0\kk$ and quit.
\item\label{st:nonsplit_sqrt:alpha_square} Check if $a\alpha$ is a square in~$K$. If there is $c\in \un$ such that $a\alpha = c^2$, then output $\qq = 0 + (\sfrac{c}{\alpha})\ii + 0\jj + 0\kk$ and quit.
\item\label{st:nonsplit_sqrt:aalpha_square} Check if $a\beta$ is a square in~$K$. If there is $c\in \un$ such that $a\beta = c^2$, then output $\qq = 0 + 0\ii + (\sfrac{c}{\beta})\jj + 0\kk$ and quit.
\item\label{st:nonsplit_sqrt:intersect} Execute Algorithm~\ref{alg:intersection} with input $\xi = \form{a, -\alpha}$ and $\zeta = \form{\beta, -\alpha\beta}$. If it fails, then report a failure and quit. Otherwise, let $d\in \un$ denote the outputted element represented by these two binary forms.
\item Construct two quadratic extensions of K:
\[
L := K\bigl(\sqrt{\alpha}\bigr)
\qquad\text{and}\qquad
M := K\bigl(\sqrt{a\alpha}\bigr).
\]
\item\label{st:nonsplit_sqrt:two_norms} Solve the following two norm equations:
\[
\frac{d}{\beta} = \norm[\ext{L}{K}]{x}
\qquad\text{and}\qquad
\frac{d}{a} = \norm[\ext{M}{K}]{y}.
\]
Denote the solutions by
\[
\lambda = l_0 + l_1\sqrt{\alpha}
\qquad\text{and}\qquad
\mu = m_0 + m_1\sqrt{a\alpha},
\]
respectively.
\item\label{st:nonsplit_sqrt:output} Output $\qq = 0 + a\cdot \frac{m_1}{m_0}\ii + \frac{l_0}{m_0}\jj + \frac{l_1}{m_0}\kk$.
\end{enumerate}
\end{alg}

\begin{poc}
The correctness of the results outputted in step~\eqref{st:nonsplit_sqrt:a_square} is  obvious as is the correctness of output of steps (\ref{st:nonsplit_sqrt:alpha_square}--\ref{st:nonsplit_sqrt:aalpha_square}). Indeed, if $a\alpha = c^2$ for some $c\in \un$ and $\qq = (\sfrac{c}{\alpha})\ii$, then $\qq^2 = \alpha\cdot \sfrac{c^2}{\alpha^2} = a$. In the remainder of the proof, we can, thus, assume that neither $a$ nor $a\alpha$ is a square in~$K$. Likewise, $\alpha$ is not a square, either, since otherwise, the quaternion algebra~$\CQ$ would split. Therefore, $L$ and $M$ are proper quadratic extensions of~$K$. It follows from Observation~\ref{obs:sqrt_in_K} that $a$ is a square of some pure quaternion $\qq = q_1\ii + q_2\jj + q_3\kk$ if and only if
\[
a\cdot 1^2 - \alpha\cdot q_1^2
=
\beta\cdot q_2^2 - \alpha\beta\cdot q_3^2.
\]
This equality is equivalent to the condition that the sets of elements of~$K$ represented by the binary forms $\xi = \form{a, -\alpha}$ and $\zeta = \form{\beta, -\alpha\beta}$ have a non-empty intersection. Thus, if Algorithm~\ref{alg:intersection} executed in step~\eqref{st:nonsplit_sqrt:intersect} reports a failure, then $a$ is not a square in~$\CQ$. Now, assume that Algorithm~\ref{alg:intersection} returned some element $d\in D_K(\xi)\cap D_K(\zeta)$. Then there are $l_0, l_1, m_0, m_1\in K$ such that
\[
\begin{cases}
d = am_0^2 - \alpha(am_1)^2 = a\cdot \norm[\ext{M}{K}]{ m_0 + m_1\sqrt{a\alpha} }\\
d = \beta l_0^2 - \alpha\beta l_1^2 = \beta\cdot \norm[\ext{L}{K}]{l_0 + l_1\sqrt{\alpha}}.
\end{cases}
\]
Rearranging the terms we have
\[
a = \alpha\Bigl( \frac{am_1}{m_0} \Bigr)^2 +  \beta\Bigl( \frac{l_0}{m_0} \Bigr)^2 -  \alpha\beta\Bigl( \frac{l_1}{m_0} \Bigr)^2.
\]
Now, the right-hand-side is nothing else but the square of the quaternion~$\qq$ constructed in step~\eqref{st:nonsplit_sqrt:output}. This proves that the algorithm is correct.
\end{poc}


\begin{thebibliography}{10}
\providecommand{\url}[1]{\texttt{#1}}
\providecommand{\urlprefix}{URL }
\expandafter\ifx\csname urlstyle\endcsname\relax
  \providecommand{\doi}[1]{doi:\discretionary{}{}{}#1}\else
  \providecommand{\doi}{doi:\discretionary{}{}{}\begingroup
  \urlstyle{rm}\Url}\fi
\providecommand{\eprint}[2][]{\url{#2}}

\bibitem{Heath1981}
Heath T.
\newblock A history of {G}reek mathematics. {V}ol. {I}.
\newblock Dover Publications, Inc., New York, 1981.
\newblock ISBN 0-486-24073-8.
\newblock From Thales to Euclid, Corrected reprint of the 1921 original.

\bibitem{Niven1942}
Niven I.
\newblock The roots of a quaternion.
\newblock \emph{Amer. Math. Monthly}, 1942.
\newblock \textbf{49}:386--388.
\newblock \doi{10.2307/2303134}.
\newblock \urlprefix\url{https://doi.org/10.2307/2303134}.

\bibitem{Vigneras1980}
Vign\'{e}ras MF.
\newblock Arithm\'{e}tique des alg\`ebres de quaternions, volume 800 of
  \emph{Lecture Notes in Mathematics}.
\newblock Springer, Berlin, 1980.
\newblock ISBN 3-540-09983-2.

\bibitem{Voight2021}
Voight J.
\newblock Quaternion algebras, volume 288 of \emph{Graduate Texts in
  Mathematics}.
\newblock Springer, Cham, [2021] \copyright 2021.
\newblock ISBN 978-3-030-56692-0; 978-3-030-56694-4.
\newblock \doi{10.1007/978-3-030-56694-4}.
\newblock \urlprefix\url{https://doi.org/10.1007/978-3-030-56694-4}.

\bibitem{Lam2005}
Lam TY.
\newblock Introduction to quadratic forms over fields, volume~67 of
  \emph{Graduate Studies in Mathematics}.
\newblock American Mathematical Society, Providence, RI, 2005.
\newblock ISBN 0-8218-1095-2.

\bibitem{GM1965}
Gordon B, Motzkin TS.
\newblock On the zeros of polynomials over division rings.
\newblock \emph{Trans. Amer. Math. Soc.}, 1965.
\newblock \textbf{116}:218--226.
\newblock \doi{10.2307/1994114}.
\newblock \urlprefix\url{https://doi.org/10.2307/1994114}.

\bibitem{Cohen2000}
Cohen H.
\newblock Advanced topics in computational number theory, volume 193 of
  \emph{Graduate Texts in Mathematics}.
\newblock Springer-Verlag, New York, 2000.
\newblock ISBN 0-387-98727-4.
\newblock \doi{10.1007/978-1-4419-8489-0}.
  \urlprefix\url{https://doi.org/10.1007/978-1-4419-8489-0}.

\bibitem{FJP1997}
Fieker C, Jurk A, Pohst M.
\newblock On solving relative norm equations in algebraic number fields.
\newblock \emph{Math. Comp.}, 1997.
\newblock \textbf{66}(217):399--410.
\newblock \doi{10.1090/S0025-5718-97-00761-8}.
  \urlprefix\url{https://doi.org/10.1090/S0025-5718-97-00761-8}.

\bibitem{FP1983}
Fincke U, Pohst M.
\newblock A procedure for determining algebraic integers of given norm.
\newblock In: Computer algebra ({L}ondon, 1983), volume 162 of \emph{Lecture
  Notes in Comput. Sci.}, pp. 194--202. Springer, Berlin, 1983.
\newblock \doi{10.1007/3-540-12868-9\_103}.
  \urlprefix\url{https://doi.org/10.1007/3-540-12868-9_103}.

\bibitem{Garbanati1980}
Garbanati DA.
\newblock An algorithm for finding an algebraic number whose norm is a given
  rational number.
\newblock \emph{J. Reine Angew. Math.}, 1980.
\newblock \textbf{316}:1--13.
\newblock \doi{10.1515/crll.1980.316.1}.
  \urlprefix\url{https://doi.org/10.1515/crll.1980.316.1}.

\bibitem{Simon2002}
Simon D.
\newblock Solving norm equations in relative number fields using {$S$}-units.
\newblock \emph{Math. Comp.}, 2002.
\newblock \textbf{71}(239):1287--1305.
\newblock \doi{10.1090/S0025-5718-02-01309-1}.
  \urlprefix\url{https://doi.org/10.1090/S0025-} \url{5718-02-01309-1}.

\bibitem{Czogala2001}
Czoga{\l}a A.
\newblock Witt rings of {H}asse domains of global fields.
\newblock \emph{J. Algebra}, 2001.
\newblock \textbf{244}(2):604--630.
\newblock \doi{10.1006/jabr.2001.8918}.
\newblock \urlprefix\url{https://doi.org/10.1006/jabr.2001.8918}.

\bibitem{CBFS2015}
Cannon J, Bosma W, Fieker C, (eds) AS.
\newblock Handbook of Magma Functions, 2.26-4 edition, 2021.

\bibitem{Koprowski2021}
Koprowski P.
\newblock Computing singular elements modulo squares.
\newblock \emph{Fund. Inform.}, 2021.
\newblock \textbf{179}(3):227--238.
\newblock \doi{10.3233/fi-2021-2022}.
  \urlprefix\url{https://doi.org/10.3233/fi-2021-2022}.

\bibitem{KR2023}
Koprowski P, Rothkegel B.
\newblock The anisotropic part of a quadratic form over a number field.
\newblock \emph{J. Symbolic Comput.}, 2023.
\newblock \textbf{115}:39--52.
\newblock \doi{10.1016/j.jsc.2022.07.003}.
\newblock \urlprefix\url{https://doi.org/10.1016/j.jsc.2022.07.003}.

\bibitem{KCz2018}
Koprowski P, Czoga{\l}a A.
\newblock Computing with quadratic forms over number fields.
\newblock \emph{J. Symbolic Comput.}, 2018.
\newblock \textbf{89}:129--145.
\newblock \doi{10.1016/j.jsc.2017.11.009}.
\newblock \urlprefix\url{https://doi.org/10.1016/j.jsc.2017.11.009}.

\bibitem{LW1992}
Leep D, Wadsworth A.
\newblock The {H}asse norm theorem mod squares.
\newblock \emph{J. Number Theory}, 1992.
\newblock \textbf{42}(3):337--348.
\newblock \doi{10.1016/0022-314X(92)90098-A}.
\newblock \urlprefix\url{https://doi.org/10.1016/0022-314X(92)90098-A}.

\end{thebibliography}


\end{document}